\documentclass[a4paper,twoside,onecolumn]{IEEEtran}
\usepackage{amsthm}
\usepackage{amsmath}
\usepackage{mathrsfs}
\usepackage{amssymb}
\usepackage{multirow}
\usepackage{units}
\usepackage{stmaryrd}
\usepackage{verbatim}
\usepackage{enumerate}
\usepackage{graphicx}
\usepackage{cite}
\usepackage{xcolor}
\usepackage[ruled,vlined, linesnumbered]{algorithm2e}

\newcommand{\bibfilePath}{../../../../ee/mybib}
\newcommand{\twobibs}[2]{#2} 

\newcommand{\xvec}{\mathbf{x}}																							
\newcommand{\boost}{\kappa}																									
\newcommand{\BoostSet}{K}																										

\newcommand{\inputAlphabet}{\mathcal{X}^t}
\newcommand{\GenAlph}{\mathcal{Y}}																					
\newcommand{\MergeAlph}{\mathcal{Z}}																				
\newcommand{\UpgAlph}{\GenAlph \cup \BoostSet}															
\newcommand{\FinAlph}{\MergeAlph \cup \BoostSet}														
\newcommand{\FinAlphShort}{\mathcal{Z}^{\,'}}																

\newcommand{\GenChannel}  {W: \inputAlphabet\rightarrow\GenAlph}    					
\newcommand{\GenC}        {W}
\newcommand{\UpgChannel}  {W^{\,'}:\inputAlphabet\rightarrow(\UpgAlph)}		 	
\newcommand{\UpgC}        {W^{\,'}}
\newcommand{\FinChannel}  {Q^{\,'}:\inputAlphabet\rightarrow(\FinAlph)}  		
\newcommand{\FinC}        {Q^{\,'}}

\newcommand{\Pz}        	{p_{\mathcal{B}}(z)}															
\newcommand{\APPz}				{\psi}																						
\newcommand{\EntropyZ}    {H_\APPz(\mathbf{X}|Z=z)}													

\theoremstyle{remark}	\newtheorem{theo}{Theorem}
\theoremstyle{remark}	\newtheorem{lemm}[theo]{Lemma}
\theoremstyle{remark}	\newtheorem{coro}[theo]{Corollary}

\newcommand{\bfX}{\mathbf{X}}

\newcommand{\myset}[1]{\left\{#1\right\}}

\SetKwInOut{Input}{input\hspace*{-1.5ex}}
\SetKwInOut{Output}{\hspace*{-1.10ex}output\hspace*{-0.35ex}}

\SetKwData{key}{key}

\DontPrintSemicolon
\SetAlgoSkip{-5pt}

\title{Channel Upgradation for Non-Binary 
 Input Alphabets and MACs}
\author{{Uzi Pereg and Ido Tal}\\
\authorblockA{
Department of Electrical Engineering,\\
Technion, Haifa 32000, Israel.\\
Email: {\tt uzipereg@tx.technion.ac.il}, {\tt idotal@ee.technion.ac.il} 
}
\thanks{The paper was presented in part at the
2014 IEEE International Symposium on Information Theory,
Honolulu, Hawaii, June 29 -- July 5, 2014.
Research supported in part by the Israel Science Foundation
grant 1769/13.}}

\begin{document}
\maketitle
\begin{abstract}
Consider a single-user or multiple-access channel with a large output alphabet. A method to approximate the channel by an upgraded version having a smaller output alphabet is presented and analyzed. The original channel is not necessarily symmetric and does not necessarily have a binary input alphabet. Also, the input distribution is not necessarily uniform. The approximation method is instrumental when constructing capacity achieving polar codes for an asymmetric channel with a non-binary input alphabet. Other settings in which the method is instrumental are the wiretap setting as well as the lossy source coding setting. 
\end{abstract}

\begin{IEEEkeywords}
Polar codes, multiple-access channel, sum-rate, asymmetric channels, 
channel degradation, channel upgradation.
\end{IEEEkeywords}

\section{Introduction}
Polar codes were introduced in 2009 in a seminal paper \cite{Arikan:09p} by Ar\i{}kan. In \cite{Arikan:09p}, Ar\i{}kan considered the case in which information is sent over a binary-input memoryless channel. The definition of polar codes was soon generalized to channels with prime input alphabet size \cite{STA:09a}. A further generalization to a polar coding scheme for a multiple-access channel (MAC) with prime input alphabet size is presented in \cite{STY:10a} and \cite{AbbeTelatar:10p}.

The communication schemes in \cite{STA:09a, STY:10a, AbbeTelatar:10p} are explicit, have efficient encoding and decoding algorithms, and achieve symmetric capacity (symmetric sum capacity in the MAC setting). However, \cite{STA:09a, STY:10a, AbbeTelatar:10p} do not discuss how an efficient construction of the underlying polar code is to be carried out. That is, no efficient method is given for finding the unfrozen synthesized channels. This question is highly relevant, since a straightforward attempt at finding the synthesized channels is intractable: the channel output alphabet size grows exponentially in the code length. The problem of constructing polar codes for these settings was discussed in \cite{TalSharovVardy:12c}, in which a degraded approximation of the synthesized channels is derived. The current paper is the natural counterpart of \cite{TalSharovVardy:12c}: here we derive an upgraded approximation.


In addition to single-user and multiple-access channels, polar codes have been used to tackle many classical information theoretic problems. Of these, we mention here three applications, and briefly explain the purpose of our results in each context. The interested reader will have no problem filling in the gaps.

First, we mention lossy source coding. Korada and Urbanke show in \cite{KoradaUrbanke:10p} a scheme by which polar codes can be used to achieve the rate distortion bound in a binary and symmetric setting. These techniques were generalized to a non-binary yet symmetric setting by Karzand and Telatar \cite{KarzandTelatar:10c}. Generalization of this result to a non-symmetric setting can be done by suitably applying the technique in \cite{HondaYamamoto:12c}. This is the technique we will use in our outline. In brief, lossy source coding for a non-symmetric and non-binary source can be carried out as follows. The test channel output corresponds to the source we want to compress, whereas the test channel input corresponds to a distorted representation of the source. The scheme applies a polar transformation on the channel input bits, and ``freezes'' (does not transmit) the transformed bits with a distribution that is almost uniform given past transformed bits. Namely, if an upgraded version of the distribution has an entropy very close to $1$, then surely the true distribution has an entropy that is as least as high.  We also mention an alternative technique of ``symmetrizing'' the channel, as described by Burshtein in \cite{Burshtein:15c}. For both methods, our method can be used to efficiently find which channels to freeze. 

A second setting where our method can be applied is coding for asymmetric channels. In \cite{HondaYamamoto:12c}, Honda and Yamamoto use the ideas developed in \cite{KoradaUrbanke:10p} in order to present a simple and elegant capacity achieving coding scheme for asymmetric memoryless channels (see also \cite{Mondelli+:14a} for a broader discussion). To use the notation in \cite{HondaYamamoto:12c}, a key part of the scheme is to transmit information $i$th synthetic channel if the entropy $H(U_i|U_0^{i-1})$ is very close to $1$ while the entropy $H(U_i|U_0^{i-1}, Y_0^{n-1})$ is very close to $0$. 
The method presented here can be used to check which indices satisfy the first condition. In addition, the method in \cite{TalSharovVardy:12c} can be used to 
check the second condition\footnote{The method in \cite{TalSharovVardy:12c} is stated with respect to a symmetric input distribution. In fact, the key result, \cite[Theorem 5]{TalSharovVardy:12c}, is easily seen to hold for non-uniform input distributions as well.}.

A third problem worth mentioning is the wiretap channel \cite{Wyner:75p}, as was done in \cite{MahdavifarVardy:11p, HofShamai:10a, ARTKS:10p, KoyluogluElGamal:10c}. There, we transmit information only over synthesized channels that are almost pure-noise channels to the wiretapper, Eve. In order to validate this property computationally, it suffices to show that an upgraded version of the synthesized channel is almost pure-noise.



The same problem we consider in this paper --- approximating a channel with an upgraded version having a prescribed output alphabet size ---  was recently considered by Ghayoori and Gulliver in \cite{GhayooriGulliver:12a}. Broadly speaking, the method presented in \cite{GhayooriGulliver:12a} builds upon the pair and triplet merging ideas presented in the context of binary channels in \cite{TalVardy:11a} and analyzed in \cite{PHTT:11c}. In \cite{GhayooriGulliver:12a}, it is stated that the resulting approximation is expected to be close to the original channel. As yet, we are not aware of an analysis making this claim precise.  In this paper, we present an alternative upgrading approximation method. Thus, with respect to our method, we are able to derive an upper bound on the gain in sum rate. The bound is given as Theorem~\ref{th:upgraded_MAC} below, and is the main analytical result of this paper.


The previous examples involved single-user channel. In fact, our method can be used in the more general setting in which a MAC is to be upgraded. Let the underlying MAC have input alphabet $\inputAlphabet$, where $t$ designates the number of users ($t=1$ if we are in fact considering a single-user channel). We would like to mention up-front that the running time of our upgradation algorithm grows very fast in $q=|\mathcal{X}|^t$. Thus, our algorithm can only be argued to be practical for small values of $q$. On a related note, we mention that a recent result \cite{Tal:15c} shows that, at least in the analogous case of degrading, this adverse effect cannot be avoided.      

This paper is written such that all the information needed in order to implement the algorithm and understand its performance is introduced first. Thus, the structure of this paper is as follows. In Section~\ref{sec:preliminaries} we set up the basic concepts and notation that will be used later on. Section~\ref{sec:binning}
describes the binning operation as it is used in our algorithm. The binning operation is a preliminary step used later on to define the upgraded channel. Section~\ref{sec:approximation} contains our approximation algorithm, as well as the statement of  Theorem~\ref{th:upgraded_MAC}. Section~\ref{sec:analysis} is devoted to proving Theorem~\ref{th:upgraded_MAC}.

\section{Preliminaries}
\label{sec:preliminaries}
\subsection{Multiple Access Channel}
Let $\GenChannel$ designate a generic $t$-user MAC, where $\mathcal{X}$ is the input alphabet of each user\footnote{Following the observation in \cite{Sasoglu:12c}, we do not constrain ourselves to an input alphabet which is of prime size.}, and $\GenAlph$ is the finite\footnote{The assumption that $\GenAlph$ is finite is only meant to make the presentation simpler. That is, our method readily generalizes to the continuous output alphabet case.} output alphabet. Denote a vector of user inputs by $\xvec\in\mathcal{X}^t$, where 
$\xvec=(x^{(l)})_{l=1}^t$. 

Our MAC is defined through the probability function $W$, where $W(y|\xvec)$ is the probability of observing the output $y$ given that the user input was $\xvec$.

\subsection{Degradation and Upgradation}
The notions of a (stochastically) degraded and upgraded MAC are defined in an analogous way to that of a degraded and upgraded single-user channel, respectively. That is, we say that a $t$-user MAC $Q:\mathcal{X}^t\rightarrow\MergeAlph$ is \emph{degraded} with respect to $\GenChannel$, if there exists a channel 
$\mathcal{P}:\GenAlph\rightarrow\mathcal{Z}$ such that for all $z\in\MergeAlph$ and $\xvec\in\mathcal{X}^t$,
\begin{equation*}
Q(z|\xvec)=\sum_{y\in\GenAlph} \GenC(y|\xvec)\cdot\mathcal{P}(z|y) \;.
\end{equation*}  
 In words, the output of $Q$ is obtained by data-processing the output of $\GenC$. We write $Q\preceq\GenC$ to denote that $Q$ is degraded with respect to $\GenC$.

Conversely, we say that a $t$-user MAC $\FinC: \mathcal{X}^t\rightarrow\FinAlphShort$ is \emph{upgraded} with respect to $\GenChannel$ if $\GenC$ is degraded with respect to $\FinC$. We denote this as $\FinC\succeq\GenC$. If $Q$ satisfies both $Q\preceq\GenC$ and $Q\succeq\GenC$, then $Q$ and $\GenC$ are said to be \emph{equivalent}. We express this by $\GenC\equiv Q$. Note that both $\preceq$ and $\succeq$ are transitive relations, and thus so is $\equiv$.

\subsection{The Sum-Rate Criterion}
\label{sub: sum_rate}
Let a $t$-user MAC $\GenChannel$ be given.
 Next, let $\bfX=(X^{(l)})_{l=1}^t$ be a random variable distributed over $\mathcal{X}^t$, not necessarily uniformly.
 Let $Y$ be the random variable one gets as the output of $W$ when the input is $\bfX$. The sum-rate of $W$ is defined as the mutual information
\[
 R(W)=I(\bfX;Y)  \;.
 \]
Note that by the data-processing inequality \cite[Theorem 2.8.1]{CoverThomas:06b}
\begin{IEEEeqnarray*}{rCl}
W \preceq Q & \Longrightarrow & R(W) \leq R(Q) \; , \\
W' \succeq Q & \Longrightarrow & R(W') \geq R(Q) \; .
\end{IEEEeqnarray*}
Thus, equivalent MACs have the same sum-rate.

In Section~\ref{sec:approximation} we show how to obtain an upgraded approximation of $\GenC$. The original MAC $\GenChannel$ is approximated by another MAC 
$\FinC: \mathcal{X}^t\rightarrow\FinAlphShort$ with a smaller output alphabet size. Then, we bound the difference (increment) in the sum-rate.   

The following lemma is a restatement of \cite[Lemma 2]{TalSharovVardy:12c}, and justifies the use of the sum-rate as the figure of merit. Informally, it states that if the difference in sum rate is small, then the difference in all other mutual informations of interest is small as well. 
(The same proof in \cite{TalSharovVardy:12c} holds for non-uniform input distribution as well.)
\begin{lemm}
\label{lemm:sumRateAsUpperBoundOnError}
Let $\GenChannel$ and $\FinC: \mathcal{X}^t\rightarrow\FinAlphShort$ be a pair of $t$-user MACs such that 
$\GenC \preceq \FinC $ and 
\begin{equation*}
R(\GenC)\geq R(\FinC)-\epsilon \,,
\end{equation*} 
where $\epsilon\geq0$. Let $\bfX$ be distributed over $\mathcal{X}^t$. Denote by $Y$ and $Z'$ the random variables one gets as the outputs of $\GenC$ and $\FinC$, respectively, when the input is $\bfX$. Let the sets $A$, $B$ be disjoint subsets of the user index set $\{ 1,2,\ldots,t \}$. Denote $\bfX_A=(X^{(l)})_{l\in A}$ and 
$\bfX_B=(X^{(l)})_{l\in B}$. Then,
\begin{equation*}
I(\bfX_A;\bfX_B,Y)\geq I(\bfX_A;\bfX_B,Z')-\epsilon \;. 
\end{equation*}
\end{lemm} 

\section{The Binning Operation}
\label{sec:binning}

\begin{figure*}[hbt]
\begin{center}
\includegraphics[scale=0.85]{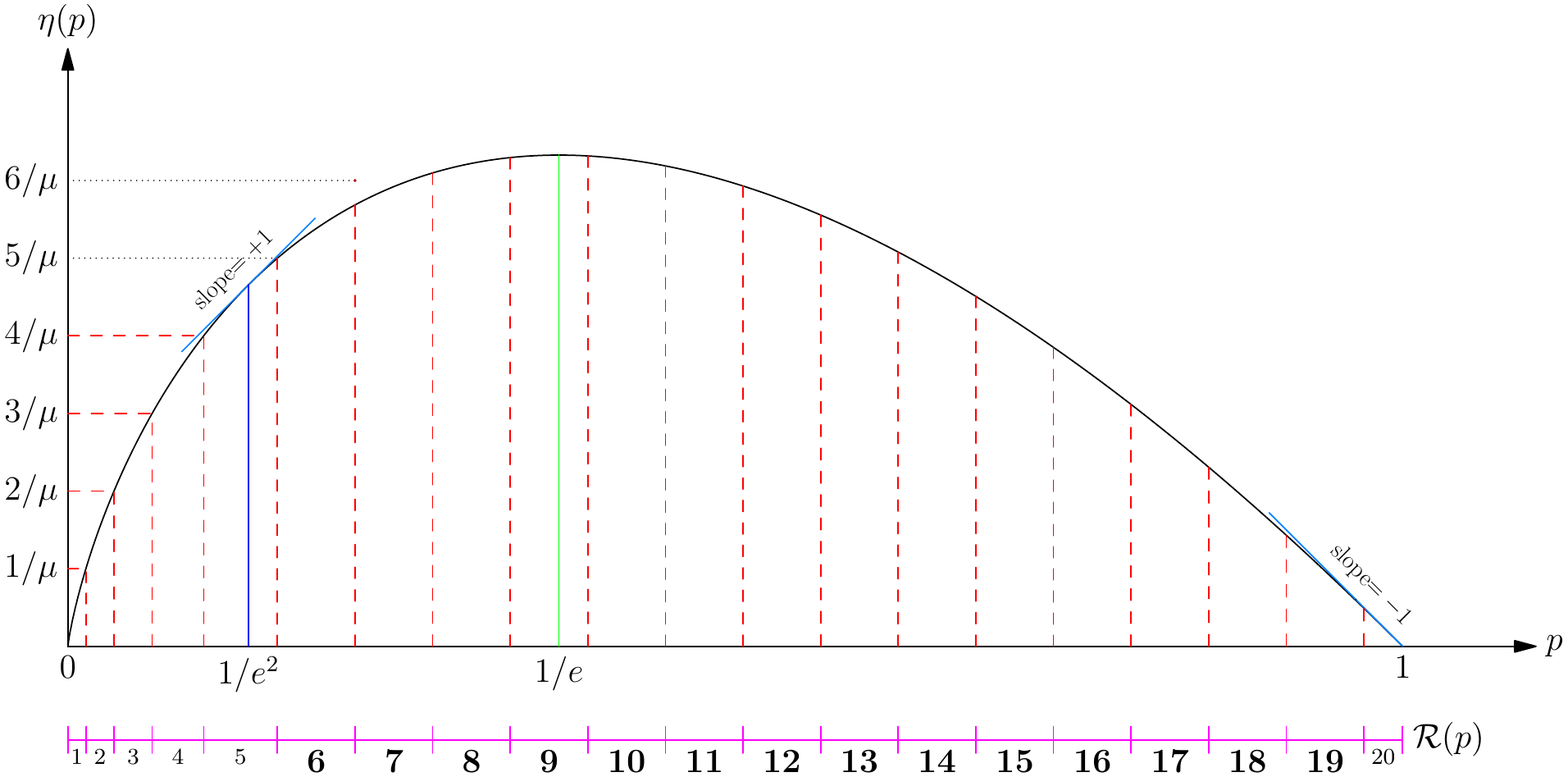} 
\caption{Functions $\eta(p) = -p\cdot \ln p$ and $\mathcal{R}(p)$. The fidelity parameter $\mu$ is set to $\mu=17.2$, which results in the number of regions being $M=20$. Some of the regions of $\mathcal{R}(x)$ are thinner (in width), and have a vertical increment of exactly $\nicefrac{1}{\mu}$. On the other hand,
the bold-faced regions have a horizontal increment (width) of exactly $\nicefrac{1}{\mu}$, while their vertical increment is less than $\nicefrac{1}{\mu}$, as the horizontal dotted lines in the figure demonstrate for region $6$. As a leftover effect, the last region, $20$, has horizontal and vertical increments which are both less than $\nicefrac{1}{\mu}$.}
\label{fig:bin}
\end{center}
\end{figure*}

\subsection{Regions and Bins}
\label{sec:bin_def}
In  \cite{TalSharovVardy:12c}, a binning operation was used to approximate a given channel by a degraded version of it. Our algorithm uses a related yet different binning rule, as a preliminary step towards upgrading the channel $\GenChannel$.

Let the random variables $\bfX$ and $Y$ be as in Lemma~\ref{lemm:sumRateAsUpperBoundOnError}, and recall that $\bfX$ is not necessarily uniformly distributed. Assume that the output alphabet $\GenAlph$ has been purged of all 
letters $y$ with zero probability of appearing under the given input distribution. That is, assume that for all $y \in \GenAlph$, the denominator in (\ref{eq:APP}) below is positive. Thus, we can indeed define the function 
$\varphi_\GenC:\mathcal{X}^t\times\GenAlph\rightarrow [0,1]$ as the a posteriori probability (APP):
\begin{equation}
 \varphi_\GenC(\xvec|y)=\mathbb{P}(\bfX=\xvec|Y=y)=\frac{\mathbb{P}(\bfX=\xvec)\cdot\GenC(y|\xvec)}
																																{\sum\limits_{\mathbf{v}\in\mathcal{X}^t} \mathbb{P}(\bfX=\mathbf{v})\cdot\GenC(y|\mathbf{v})} \;, 
\label{eq:APP}
\end{equation}
for every input $\xvec\in\mathcal{X}^t$ and every letter in the (purged) output alphabet $y\in\GenAlph$. Next, for $y\in\GenAlph$ let us denote 
\[ 
p_\GenC(y)=\mathbb{P}(Y=y) \;,
 \]
and define $\eta:[0,1]\rightarrow\mathbb{R}\,$ by
\[ 
\eta(p)=-p\cdot\ln p \;,
 \]
where $\ln (\cdot)$ stands for \emph{natural} logarithm.
Using the above notation, the entropy of the input $\bfX$ given the observation $Y=y$ is
\begin{IEEEeqnarray*}{l}
H(\bfX|Y=y)= \sum_{\xvec\in\mathcal{X}^t} \eta\left(\varphi_\GenC(\xvec|y)\right) \;,
\end{IEEEeqnarray*}
measured in natural units (nats). Thus, the sum-rate can be expressed as 
\begin{IEEEeqnarray*}{rCl}
R(W)
&=& H(\bfX) -\sum_{y\in\GenAlph} p_\GenC(y) H(\bfX|Y=y) \\  
&=& H(\bfX) -\sum_{y\in\GenAlph} p_\GenC(y)\sum_{\xvec\in\mathcal{X}^t} \eta\left(\varphi_\GenC(\xvec|y)\right) \;.
\end{IEEEeqnarray*}

As a first step towards the definition of our bins, we quantize the domain of $\eta(p)$  with resolution specified by a fidelity parameter $\mu$. 
That is, we partition $[0,1]$ into quantization-regions which depend on the value of $\mu$. 
Informally, we enlarge the width of each region until an increment of 
$\nicefrac{1}{\mu}$ is reached, either on the horizontal or the vertical axis. To be exact, the interval $[0,1)$ is partitioned into $M=M_\mu$ non-empty regions of the form 
\[ 
[b_i,b_{i+1}) \quad,\quad i=1,2,\ldots,M \;.
 \]
Starting from $b_1=0$, the endpoint of the $i$th region is given by
\begin{IEEEeqnarray}{l}
\label{eq:quantization}
 b_{i+1}= \max \left\{\;\; 0< p\leq1 \;:\quad
 x\leq b_i+\frac{1}{\mu} \;,\quad 
\quad|\eta(p)-\eta(b_i)|\leq\frac{1}{\mu} 
\;\;\right\}\;. 
\end{IEEEeqnarray}
And so it is easily inferred that for all regions $1\leq i< M$ (all regions but the last),  there is \emph{either} a horizontal \emph{or} vertical increment 
of $\nicefrac{1}{\mu}$: 
\begin{IEEEeqnarray*}{l}
 b_{i+1}-b_i=\frac{1}{\mu} \;,\quad \text{or}\quad |\eta(b_{i+1})-\eta(b_i)|=\frac{1}{\mu} \;,
\end{IEEEeqnarray*}
but typically not both (Figure~\ref{fig:bin}). For technical reasons, we will henceforth assume that
\begin{equation}
\label{eq:mu_req}
\mu \geq \max(5,q(q-1)) \; . 
\end{equation}

Denote the region to which $x$ belongs by $\mathcal{R}(x)=\mathcal{R}_\mu(x)$. Namely,
\begin{subequations}
\label{eq:regionDefintions}
\begin{equation}
\label{seq:regionGeneralDefintion}
\mathcal{R}(x)=i \quad\Leftrightarrow\quad x\in[b_i,b_{i+1})\;,
\end{equation}
with the exception of $x=1$ belonging to the last region, meaning
\begin{equation}
\label{seq:regionSpecialCase}
\mathcal{R}(1)=M \; .
\end{equation}
\end{subequations}
Based on the quantization regions defined above, we define our binning rule.
Two output letters $y_1,y_2 \in  \GenAlph$ are said to be in the same bin if for all $\mathbf{u}\in\mathcal{X}^t$ we have that 
$\mathcal{R}(\varphi_\GenC(\xvec|y_1))=\mathcal{R}(\varphi_\GenC(\xvec|y_2))$. That is, $y_1$ and $y_2$ share
the same vector of region-indices,
\[
\left(\, i(\xvec) \,\right)_{\,\xvec\in\mathcal{X}^t} \quad,
\]
where $i(\xvec)\triangleq\mathcal{R}(\varphi_\GenC(\xvec|y_1))=\mathcal{R}(\varphi_\GenC(\xvec|y_2))$.  
Note that we will try to use consistent terminology throughout: A ``region'' is an one-dimensional interval and has to do with a specific value of $\xvec$. A ``bin'' is essentially a $q$-dimensional cube, defined through regions, and has to do with all the values $\xvec$ can take.

\subsection{Merging of letters in the same bin} 
\label{sec:Bins_Merge}
Recall that our ultimate aim is to approximate the original channel $\GenChannel$ by an upgraded version having a smaller output alphabet. As we will see, the output alphabet of the approximating channel will be a union of two sets. In this subsection, we define one of these sets, denoted by $\MergeAlph$. 

Figuratively, we think of $\MergeAlph$ as the result of merging together all the letters in the same bin. That is, the size of $\MergeAlph$ is the number of non-empty bins, as each non-empty bin corresponds to a distinct letter $z\in\MergeAlph$. Denote by $\mathcal{B}(z)$ the set of letters in $\GenAlph$ which form the bin associated with $z$.  Thus, all the symbols $y\in\mathcal{B}(z)$ can be thought of as having been merged into one symbol $z$.

As we will see, the size of $\MergeAlph$ can be upper-bounded by an expression that is not a function of $|\GenAlph|$.

\subsection{The APP measure $\APPz$}
\label{subsec:psiSpecification}
In this subsection, we define an a posteriori probability measure on the input alphabet $\inputAlphabet$, given a letter from the merged output alphabet $\MergeAlph$. We denote this APP measure as $\APPz(\xvec|z)$, defined for $\xvec \in \inputAlphabet$ and $z \in \MergeAlph$.

The measure $\APPz(\xvec|z)$ will be used in Section~\ref{sec:approximation} in order to define the approximating channel. As we have previously mentioned, the output alphabet of the approximating channel will contain $\MergeAlph$. As we will see, $\APPz(\xvec|z)$ will equal the APP of the approximating channel, for output letters 
$z\in\MergeAlph$.


For each bin define the \emph{leading input} as 
\begin{IEEEeqnarray}{l}
\label{eq:leadingu}
 \xvec^*=\xvec^*(z) \triangleq \arg\max_{\xvec\in\mathcal{X}^t} \; \left[\, \max_{ y\in\mathcal{B}(z)} \varphi_\GenC(\xvec|y) \,\right] \; ,
\label{eq:leading_input}
\end{IEEEeqnarray}
where ties are broken arbitrarily. For $z\in\MergeAlph$, let 
\begin{subequations}
\label{eq:Merge_APP}
\begin{equation}
\label{eq:Merge_App:notLeading}
\APPz(\xvec|z)= \min_{y\in\mathcal{B}(z)} \,\varphi_\GenC(\xvec|y) \quad \text{for all } \xvec\neq\xvec^*,
\end{equation}
and
\begin{equation}
\label{eq:Merge_App:leading}
\APPz(\xvec^*|z)= 1-\sum_{\xvec\neq\xvec^*} \APPz(\xvec|z)  \;.\qquad\qquad\quad\, 
\end{equation} 
\end{subequations}

Informally, we note that if the bins are ``sufficiently narrow'' (if $\mu$ is sufficiently large), then $\APPz(\xvec|z)$ is close to $\varphi_\GenC(\xvec|y)$, for all  $\xvec \in \inputAlphabet$, $z \in \MergeAlph$, and $y\in\mathcal{B}(z)$. The above will be made exact in Lemma~\ref{lemma:eta_bound} below.

\section{The Upgraded Approximation}
\label{sec:approximation}
\subsection{Definition}
Now we are in position to define our $t$-user MAC approximation $\FinChannel$, where $\BoostSet$ is a set of additional symbols to be specified in this section. We refer to these new symbols as ``boost'' symbols.

Let $y\in\GenAlph$ and $\xvec\in\mathcal{X}^t$ be given, and let $z$ correspond to the bin $\mathcal{B}(z)$ which contains $y$. Define the quantity $\alpha_\xvec(y)$ as
\begin{subequations}
\label{eq:alpha}
\begin{equation}
\label{eq:alpha_general}
\alpha_\xvec (y) \triangleq \frac{\APPz(\xvec|z)}{\varphi_\GenC(\xvec|y)} \cdot \frac{\varphi_\GenC(\xvec^*|y)}{\APPz(\xvec^*|z)}   \;,\quad 
\text{if }\; \varphi_\GenC(\xvec|y) \neq 0 \;.
\end{equation}
Otherwise, define
\begin{equation}
\label{eq:alpha_exception}
\alpha_\xvec (y) \triangleq 1 \;,\quad 
\text{if }\; \varphi_\GenC(\xvec|y) = 0  \;. 
\end{equation}
\end{subequations}
By Lemma~\ref{lemm:alpha} in the next section, 
$\alpha_\xvec (y)$ is indeed well defined and is between $0$ and $1$.
Next, for $\xvec\in\mathcal{X}^t$, let
\begin{equation} 
\varepsilon_\xvec\triangleq \sum\limits_{y\in\GenAlph}^{} (1-\alpha_\xvec (y))\GenC(y|\xvec)  \;.\label{eq:epsilon}
\end{equation}
We now define $\BoostSet$, the set of output ``boost'' symbols. Namely, we define a boost symbol for each non-zero $\varepsilon_\xvec$,
\begin{equation}
\BoostSet = \{\, \boost_\xvec \,:\, \xvec\in\mathcal{X}^t \,,\, \varepsilon_\xvec>0 \,\} \; . \label{eq:BoostSet}
\end{equation}

Lastly, the probability function $\FinC$ of our upgraded MAC is defined as follows. With respect to non-boost symbols, define for all $z\in\MergeAlph$ and $\xvec\in\mathcal{X}^t$, 
\begin{subequations}
\label{eq:FinC}
\begin{equation}
 \FinC(z|\xvec)= \sum_{y\in\mathcal{B}(z)} \alpha_\xvec(y)\GenC(y|\xvec) \;. 
\label{eq:FinC_1}
\end{equation}
With respect to boost symbols, define for all $\boost_\mathbf{v}\in\BoostSet$ and  $\xvec\in\mathcal{X}^t$, 
\begin{equation}
\FinC(\boost_\mathbf{v}|\xvec)=
\begin{cases}
\varepsilon_\xvec& \mbox{if $\xvec=\mathbf{v}$} \;, \\
0	& \text{otherwise} \; .	
\end{cases}
\label{eq:FinC_2}
\end{equation}
\end{subequations}
Note that if a boost symbol  $\boost_\xvec$ is received at the output of $\FinChannel$, we know for certain that the input was $\bfX=\xvec$.

The following theorem presents the properties of our upgraded approximation of $\GenC$. The proof concludes 
Section~\ref{sec:analysis}.
\begin{theo}
\label{th:upgraded_MAC}
Let $\GenChannel$ be a $t$-user MAC, and let $\mu$ be a given fidelity parameter that satisfies (\ref{eq:mu_req}) . Let $\FinChannel$ be the MAC obtained from $\GenC$ by the above definition (\ref{eq:FinC}).
Then,
\begin{enumerate}[(i)]
	\item \label{th:upgraded_MAC:upgraded} The MAC $\FinC$ is well defined and is upgraded with respect to $\GenC$.
	\item \label{th:upgraded_MAC:sumRateBound} The increment in sum-rate is bounded by 
	\begin{IEEEeqnarray*}{l}
	R(\FinC)-R(\GenC) \leq  \frac{q-1}{\mu}\left( 2+q\cdot\ln q  \right) \;.  
	\end{IEEEeqnarray*}
	\item \label{th:upgraded_MAC:outputAlphabetBound} The output alphabet size of $\FinC$ is bounded by $q^2 \cdot  (2\mu)^{q-1}$. 
\end{enumerate}
\end{theo}

Note that the input alphabet size $q$ is usually considered to be a given parameter of the communications system. Therefore, we can think of $q$ as being a constant. In this view, 
Theorem~\ref{th:upgraded_MAC} claims that our upgraded-approximation has a sum-rate deviation of $\mathcal{O}(\frac{1}{\mu})$, and an output-alphabet of size $\mathcal{O}(\mu^{q-1})$.  

\subsection{Implementation}
In this subsection, we outline an efficient implementation of our algorithm. In short, we make use of an associative array, also called a dictionary \cite[Page 197]{CLRS:01b}. An associative array is a data structure through which elements can be searched for by a key, accessed, and iterated over efficiently. In our case, the elements are sets, and they are represented via linked lists \cite[Subsection 10.2]{CLRS:01b}. The associative array can be implemented as a self-balancing tree \cite[Section 13]{CLRS:01b} holding (pointers to) the lists. A different approach is to implement the associative array as a dynamically growing \cite[Subsection 17.4]{CLRS:01b} hash table \cite[Subsection 11.2]{CLRS:01b}. Algorithm~\ref{alg:Upgrade} summarizes our implementation.

We draw the reader's attention to the following. Consider the variables $\alpha_\xvec (y)$ and $\APPz(\xvec|z)$ used in the algorithm. The naming of these variables is meant to be consistent with the other parts of the paper. However, note that there are in fact only two floating point variables involved. That is, once we have finished dealing with $y_1$ and moved on to dealing with $y_2$ in the innermost loop on line~\ref{algline:innermostLoop}, the memory space used in order to hold $\alpha_\xvec (y_1)$ should be reused in order to hold $\alpha_\xvec (y_2)$, etc.

\begin{algorithm}
\caption{Channel upgrading procedure}
\label{alg:Upgrade}
\Input{A MAC $\GenC : \inputAlphabet \to \GenAlph$, a fidelity parameter $\mu$.}
\Output{A MAC $\FinChannel$ satisfying Theorem~\ref{th:upgraded_MAC}.}
\tcp{Initialization of region boundaries}
Calculate the number of regions $M$ and region boundaries $(b_i)_{i=0}^M$ according to (\ref{eq:quantization}) \; \label{algline:calc_bi}
\tcp{Initialization of data structure}
Initialize an empty associative array (containing no lists) \; 
\tcp{Populate the data structure}
\For{each $y \in \GenAlph$}
{
  \tcp{Calculate key according to (\ref{eq:APP}) and (\ref{eq:regionDefintions})}
  $\key = \left(\, i(\xvec) \,\right)_{\,\xvec\in\mathcal{X}^t} \; , \mbox{where} \quad i(\xvec)\triangleq\mathcal{R}(\varphi_\GenC(\xvec|y))$ \; \nllabel{algline:keyDeff}
  \eIf{associative array contains a linked list corresponding to $\key$}  
  { \nllabel{algline:is_key_in_list}
    add $y$ to the corresponding linked list
  }
  {
    create a new (empty) linked list, add $y$ to it, add the linked list to the associative array by associating it with $\key$
  }
}
\tcp{Initialize $\varepsilon_\xvec$}
Set $\varepsilon_\xvec = 0$, for each $\xvec \in \inputAlphabet$ \; 
\tcp{Iterate over all non-empty bins}
\tcp{Produce non-boost symbols and probabilities}
\For{each linked list in the associative array}
{
Create a new letter $z$ and add it to the output alphabet of $\FinC$ \;
Set $\FinC(z|\xvec) = 0$, for each $\xvec \in \inputAlphabet$ \; 
Loop over all $y$ in list and all  $\xvec \in \inputAlphabet$. Calculate the leading input $\xvec^*$ according to (\ref{eq:leadingu}) \;
\For{each $\xvec \in \inputAlphabet$}
{
Loop over all $y$ in list and calculate $\APPz(\xvec|z)$ according to (\ref{eq:Merge_APP}) \;
\For{each $y$ in the linked list}
{ \nllabel{algline:innermostLoop}
  Calculate $\alpha_\xvec(y)$ according to (\ref{eq:APP}) and (\ref{eq:alpha}) \;
  \tcp{Implement (\ref{eq:epsilon})}
  Increment $\varepsilon_\xvec$ by $(1-\alpha_\xvec (y))\GenC(y|\xvec)$ \; 
  \tcp{Implement (\ref{eq:FinC_1})}
  Increment $\FinC(z|\xvec)$ by $\alpha_\xvec(y)\GenC(y|\xvec)$
}
}
}
\tcp{Produce boost symbols and probabilities }
\For{each $\mathbf{v} \in \inputAlphabet$}
{
\If{$\varepsilon_\mathbf{v} > 0$}
{
Create a new letter $\boost_\mathbf{v}$ and add it to the output alphabet of $\FinC$ \;
\tcp{Implement (\ref{eq:FinC_2})}
\For{each $\xvec \in \inputAlphabet$}
{
\eIf{$\xvec = \mathbf{v}$}
{
Set $\FinC(\boost_\mathbf{v}|\xvec)=\varepsilon_\mathbf{v}$
}
{
Set $\FinC(\boost_\mathbf{v}|\xvec)=0$
}
}
}
}
\end{algorithm}

Let us now analyze our algorithm. Consider first the time complexity. We will henceforth assume that the total number of regions, $M$, is less than the largest integer value supported by the computer. We will further assume that integer operations are carried out in time $O(1)$. Hence, the calculation of a key takes time $O(q \cdot \log M)$. To see this, first recall that by line~\ref{algline:keyDeff} of the algorithm, a key is simply a vector of length $q$ containing region indices. Finding the correct region index for each value of $\xvec$ can be done by a binary search involving the $b_i$ calculated in line~\ref{algline:calc_bi}. Since line~\ref{algline:keyDeff} is invoked $|\GenAlph|$ times, the total time spent running it is $O(|\GenAlph| \cdot q \cdot \log M)$.

We next consider the running time of a single invocation of line~\ref{algline:is_key_in_list}. Checking for key equality and order takes time $O(q)$. If a balanced tree with $n$ elements is used, this operation occurs $O(\log n)$ times for each search operation. In contrast, in a dynamic hash implementation, checking for key equality occurs only $O(1)$ times on average, for each search operation. We again recall that line~\ref{algline:is_key_in_list} is invoked $O(|\GenAlph|)$ times. Thus, the total time spent running line~\ref{algline:is_key_in_list} in the balanced tree implementation is $O(|\GenAlph| \cdot q \cdot \log n)$, where $n$ is the number of non-empty bins. In contrast, in a dynamic hash implementation, the total time spent running line~\ref{algline:is_key_in_list} is $O(|\GenAlph| \cdot q)$, on average.

By inspection, the total time spent running any other line in the algorithm is upper bounded --- up to multiplicative constants --- by the total spent running either line~\ref{algline:keyDeff} or line~\ref{algline:is_key_in_list}. Consider first the balanced tree implementation. We conclude that the running time is $O(|\GenAlph| \cdot q \cdot (\log n + \log M))$, where $n$ is the total number of non-empty bins and $M$ is the total number of regions. By Corollary~\ref{cor:intervals} below, we can write this as $O(|\GenAlph| \cdot q \cdot (\log n + \log \mu))$, where $\mu$ is the fidelity parameter. Obviously, the total number of non-empty bins is at most $|\GenAlph|$. Thus, the total running time is $O(|\GenAlph| \cdot q \cdot (\log |\GenAlph| + \log \mu))$, for the balanced tree implementation (worst case). In the hash setting, the same arguments lead us to conclude that the total running time is $O(|\GenAlph| \cdot q \cdot \log \mu)$, on average.

The space complexity of our algorithm is $O(|\GenAlph| + n \cdot q)$: we must store all the elements of $\GenAlph$, and the key corresponding to every non-empty bin. As before, we can thus bound the space complexity as $O(|\GenAlph|(q + 1))$.
\section{Analysis}
\label{sec:analysis}
Conceptually, for the purpose of analysis, the algorithm can be thought of as performing four steps. In the first step, an output alphabet $\MergeAlph$ is defined (Subsection~\ref{sec:Bins_Merge}) by means of a quantization (Subsection~\ref{sec:bin_def}). In the second step, a corresponding APP measure $\APPz$  is defined (Subsection~\ref{subsec:psiSpecification}). In the third step, the original output alphabet $\GenAlph$ is augmented with ``boost'' symbols $\BoostSet$, and a new channel $\UpgChannel$ is defined. The APP measure $\APPz$ has a key role in defining $\UpgC$, which is upgraded with respect to $\GenC$. In the fourth step, we consolidate equivalent symbols in $\UpgChannel$ into a single symbol. The resulting channel is $\FinChannel$. On the one hand, $\FinC$ is equivalent to $\UpgC$, and thus upgraded with respect to the original channel $\GenC$. On the other hand, the output alphabet of $\FinC$ turns out to be $\MergeAlph \cup \BoostSet$, a set typically much smaller than the original output alphabet $\GenAlph$. The channels used throughout the analysis are depicted in Figure~\ref{fig: MACs}, along with the corresponding properties and the relations between them.

\begin{figure*}[hbt]
\begin{center}
\begin{tabular}{lcccccc}
&&\footnotesize{ }&\footnotesize{Upgrade}&&\footnotesize{Consolidate}&\\ 
\multicolumn{7}{c}{ }\\
Channel		&			$\GenC$								 &\footnotesize{ }											&$\preceq$&			$\UpgC$								 &$\equiv$ &			$\FinC$								\\

\multicolumn{7}{c}{ }\\
Output Alphabet		&
										$\GenAlph$	&	&	&
										$\UpgAlph$		& &
										$\FinAlph$	\\
\multicolumn{7}{c}{ }\\
\hline												
\multicolumn{7}{c}{ }\\
Bottom line:			&     \multicolumn{6}{l}{ $\quad\GenC\preceq\FinC$     } \\
					\multicolumn{7}{c}{}																 \\ 
\end{tabular}
\caption{A high-level view of the MACs used throughout the analysis. 
}
\label{fig: MACs} 
\end{center}
\end{figure*}

 We now examine the algorithm step by step, and state the relevant lemmas and properties for each step. This eventually leads up to the proof of Theorem~\ref{th:upgraded_MAC}.

\subsection{Quantization Properties}
In Section~\ref{sec:bin_def}, we have quantized the domain of the function $\eta(p)=-p\cdot\ln p$ for the purpose of binning.
Now, we would like to discuss a few properties of this definition.

Observing Figure~\ref{fig:bin}, the reader may have noticed that regions entirely to the left of $x=\frac{1}{e^2}$ have a \emph{vertical} increment of $\frac{1}{\mu}$. On the other hand, regions entirely to the right of $x=\frac{1}{e^2}$, last region excluded, have a \emph{horizontal} width of $\frac{1}{\mu}$. The following lemma shows that this is always the case. 
\begin{lemm}
\label{lemma:quantization}
Let the extreme points $\{ b_i:\, 1\leq i\leq M+1 \}$ partition the domain interval $0\leq x\leq 1$ into quantization
regions (intervals), as in Section~\ref{sec:bin_def} (see (\ref{eq:quantization})). 
Thus,
\begin{enumerate}[(i)]
\item\label{lemma:quantization:it:firstCase} if $0\leq b_{i}<b_{i+1}<\frac{1}{e^2}\;$, then
 \[
\eta(b_{i+1})-\eta(b_i)=\frac{1}{\mu}\;.
\]
\item\label{lemma:quantization:it:otherwise} Otherwise, if $\frac{1}{e^2}\leq b_{i}<b_{i+1}< 1\;$, then 
\[
b_{i+1}-b_i=\frac{1}{\mu}\;.
\]
\end{enumerate}
\end{lemm}
\begin{proof} 
The derivative $\eta'(p)=-(1+\ln p)$ is strictly decreasing from $+\infty$ at $p=0\,$, to $+1$ at $p=\frac{1}{e^2}\,$. 
Thus, for all $0\leq p<\frac{1}{e^2}$,
\[
 \eta'(p)>1 \;.
 \]   
If $0\leq b_{i}<b_{i+1}<\frac{1}{e^2}\;$, then we have by the fundamental theorem of calculus that
\[
 \eta(b_i+\frac{1}{\mu})-\eta(b_i)=\int_{b_i}^{b_i+\frac{1}{\mu}} \eta'(p)\,dp > \frac{1}{\mu} \;. 
\]
Hence $b_{i+1}<b_i+\frac{1}{\mu}\,$, which implies the first part of the lemma.

Moving forward on the $x$-axis, $\eta'(p)$ keeps decreasing from $+1$ at $p=\frac{1}{e^2}\,$, to $-1$ at $p=1$.
Thus for all $\frac{1}{e^2}\leq p\leq 1$, 
\[
|\eta'(p)|\leq 1 \;.
\]
Hence, if $\frac{1}{e^2}\leq b_{i}<b_{i+1}< 1\;$,  the second part follows by the triangle inequality:
\[
|\eta(b_i+\frac{1}{\mu})-\eta(b_i)|\leq \int_{b_i}^{b_i+\frac{1}{\mu}} |\eta'(p)|\,dp \leq \frac{1}{\mu}.
 \]
\end{proof}

We are now ready to upper-bound $M = M_\mu$, the number of quantization regions. The following corollary will be used to bound the number of bins, namely $|\MergeAlph|$, later on.

\begin{coro}
\label{cor:intervals}
The number of quantization regions, $M = M_\mu$, satisfies 
\[
M\leq 2\mu \; .
\]
\end{coro}
\begin{proof}
A direct consequence of Lemma~\ref{lemma:quantization} is that
\[
 M\leq \left\lfloor \frac{\eta(\frac{1}{e^2})}{\nicefrac{1}{\mu}} \right\rfloor + \left( 
\left\lfloor \frac{1-\frac{1}{e^2}}{\nicefrac{1}{\mu}} \right\rfloor +1 
\right)+1 \;.
 \]
The first term is due to regions entirely within $[0,\frac{1}{e^2})$, the second (braced) term is due to regions entirely within $[\frac{1}{e^2},1]$, where the $1$ inside the braces is due to the last (rightmost) region. The $1$ outside the brace is due to the possibility of a region that crosses $x=\frac{1}{e^2}$. Hence, since $\eta(1/e^2) = 2/e^2$,
\begin{IEEEeqnarray*}{l}
M \leq \mu \left( 1+\frac{1}{e^2} \right)+2 \leq 2\mu\;,
\end{IEEEeqnarray*}
where the last inequality follows from our assumption in (\ref{eq:mu_req}) that  $\mu\geq5$. 
\end{proof}

The corollary, following the lemma below, will play a significant role in the proof of Theorem~\ref{th:upgraded_MAC}. The lemma is proved in the appendix.
\begin{lemm}
\label{lemma:eta_diff}
Given $x\in [0,1)$, let $i=\mathcal{R}(x)$. That is,
\[
 b_i\leq x< b_{i+1} \;. 
\]
Also, let
\[
 0< \delta\leq b_{i+1}-b_i \;,
 \]
such that $x+\delta \leq 1$.
Then,
\[ 
|\eta(p+\delta)-\eta(p)|\leq \frac{1}{\mu} 
\]
\end{lemm}
The corollary below is an immediate consequence of Lemma~\ref{lemma:eta_diff}.
\begin{coro}
\label{cor:eta_diff_inBin}
All $x_1$ and $x_2$ that belong to the same quantization region
$\left( \text{that is: } \mathcal{R}(x_1)=\mathcal{R}(x_2)\right) $
satisfy
\[
 |\eta(p_1)-\eta(p_2)|\leq\frac{1}{\mu} \;. 
\]
\end{coro}

The following lemma claims that each quantization interval, save the last, is at least as wide as the previous intervals. This lemma is proved in the appendix as well.
\begin{lemm}
\label{lemma:interval_width}
Let the width of the $i$th quantization interval be denoted by
\[
 \Delta_i=b_{i+1}-b_i \;,\quad i=1,2,\ldots,M \;.
 \]
Then the sequence $\{\Delta_i\}_{i=1}^{M-1}$ (the last interval excluded) is a non-decreasing sequence.
\end{lemm}

Following the quantization definition, the output letters in $\GenAlph$ were divided into bins (Section~\ref{sec:Bins_Merge}). Each bin is represented by a single letter in $\MergeAlph$. The following lemma upper bounds the size of $\MergeAlph$.

\begin{lemm}
\label{lemma:num_bin}
Let $\MergeAlph$ be defined as in Section~\ref{sec:Bins_Merge}. Then,
\[ 
|\MergeAlph|\leq q^2 \cdot (2\mu)^{q-1} \;. 
\]
\end{lemm}

Before stating the proof, we would like to mention that it is generic, in the following sense: the proof can be used verbatim to prove that the output alphabet size in the degrading algorithm presented in \cite{TalSharovVardy:12c} produces a channel with output alphabet size at most $q^2 \cdot (2\mu)^{q-1}$. This is an improvement over the $(2\mu)^q$ bound stated in \cite[Lemma 6]{TalSharovVardy:12c}.

\begin{proof}
The size of the merged output alphabet $|\MergeAlph|$ is in fact the number of non-empty bins. Recall that two letters $y_1,y_2\in\GenAlph$ are in the same bin if and only if $\mathcal{R}(\varphi_\GenC(\xvec|y_1))=\mathcal{R}(\varphi_\GenC(\xvec|y_2))\;$ for all 
$\xvec\in\mathcal{X}^t$. As before, denote by $M = M_\mu$ the number of quantization regions. Since the number of values $\xvec$ can take is $q$, we trivially have that
\[ 
|\MergeAlph|\leq M^q \;.
\]

We next sharpen the above bound by showing that although $M^q$ bins exist, some are necessarily empty. If a bin is non-empty, there must exist a $y \in \GenAlph$ such that $(\varphi_\GenC(\xvec|y))_{\xvec \in \inputAlphabet}$ is mapped to it. Thus, let us bound the number of valid bins, where a bin is valid if there exists a probability vector $(p[\xvec])_{\xvec \in \inputAlphabet}$ that is mapped to it. First, recall that a bin is simply an ordered collection of regions. That is, recall that for each $\xvec \in \inputAlphabet$,  $p[\xvec]$ must belong to a region of the form $[b_i,b_{i+1})$ or, if $i=M$, $[b_i,b_{i+1}]$. Thus, denote by $\underline{b}[\xvec] = b_i$ and $\overline{b}[\xvec] = b_{i+1}$ the left and right borders of this region. Let the ``widest $\xvec$'' be the $\xvec \in \inputAlphabet$ for which $\overline{b}[\xvec] - \underline{b}[\xvec]$ is largest (brake ties according to some ordering of $\inputAlphabet$, say). 

For ease of exposition, let us abuse notation and label the elements of $\inputAlphabet$ as $0,1,\ldots,q-1$. We now aim to bound the number of valid bins for which the widest $\xvec$ is $0$. Surely, there are at most $M^{q-1}$ choices for the regions corresponding to the $\xvec$ from $1$ to $q-1$. We now fix such a choice, and bound the number of regions which can correspond to $\xvec = 0$. By the above definitions, a corresponding probability vector $(p[\xvec])_{\xvec \in \inputAlphabet}$ must satisfy

\[
p[0] + \sum_{\xvec = 1}^{q-1} \underline{b}[\xvec] \leq 1 \qquad \mbox{and}  \qquad p[0] + \sum_{\xvec  = 1}^{q-1} \overline{b}[\xvec] \geq 1 \; .
\]
Thus,
\begin{equation}
\label{eq:bunderlineOverlineConstraint}
\underline{\beta} \triangleq \max\myset{0, 1 - \sum_{\xvec = 1}^{q-1} \overline{b}[\xvec]} \leq
p[0] \leq \min\myset{1, 1 - \sum_{\xvec = 1}^{q-1} \underline{b}[\xvec]} \triangleq \overline{\beta} \; .
\end{equation}

We now use the fact that $\xvec = 0$ is widest. Denote
\[
\overline{\Delta} = \max_{1 \leq \xvec \leq q-1} \myset{ \overline{b}[\xvec] - \underline{b}[\xvec]} \; .
\]
On the one hand, $p[0]$ must belong to a region with width at least $\overline{\Delta}$. On the other hand, $\overline{\beta} - \underline{\beta} \leq (q-1) \overline{\Delta}$. Thus, the number of such regions which have a non-empty intersection with the interval $[\underline{\beta},\overline{\beta}]$ is at most $q$. 

To sum up, we have shown that if the widest $\xvec$ is $0$, the number of valid bins is at most $q \cdot M^{q-1}$. Since there is no significance to the choice $\xvec = 0$, the total number of valid bins is at most $q^2 \cdot M^{q-1}$. The proof now follows from Corollary~\ref{cor:intervals}.
\end{proof}

Consider a given bin (and a given $z\in\MergeAlph$). Depending on $\xvec\in\mathcal{X}^t$, all $y\in\mathcal{B}(z)$ share the same region index
\begin{equation}
\label{eq:region_index}
 i(\xvec)=i_z(\xvec) \triangleq \mathcal{R}\left( \varphi_\GenC(\xvec|y)  \right) \;. 
\end{equation}
Denote the set of region indices associated with a bin as
\begin{equation}
\label{eq:Lz}
 \mathcal{L}(z)=\left\{ \, i_z(\xvec):\, \xvec\in\mathcal{X}^t  \, \right\} \;.
\end{equation}
According to the following lemma, the largest index in $\mathcal{L}(z)$ belongs to the leading input $\xvec^*$, defined in (\ref{eq:leadingu}). In other words the leading input is in the \emph{leading region}.

\begin{lemm}
\label{lemma:leading_region}
Consider a given $z\in\MergeAlph$. Let $i(\xvec)$ be given by (\ref{eq:region_index}) for all $\xvec\in\mathcal{X}^t$, and let $\xvec^*$ be as in (\ref{eq:leadingu}). Then
\[
 i(\xvec^*)= \max\, \left\{ \, i(\xvec):\, \xvec\in\mathcal{X}^t  \, \right\} \;.
 \]
\end{lemm}

\begin{proof}
Define the \emph{leading output} $y^*\in\mathcal{B}(z)$ by
\begin{equation}
\label{eq:yleading}
y^*=y^*(z) \triangleq \arg\max_{y\in\mathcal{B}(z)} \varphi_\GenC(\xvec^*|y) \; . 
\end{equation}
By (\ref{eq:leadingu}) and (\ref{eq:yleading}), we have that
\begin{equation}
\label{eq:phiDoubleStar}
\varphi_\GenC(\xvec^*|y^*) = \max_{\substack{\xvec \in \mathcal{X}^t \\ y \in \mathcal{B}(z)}}  \varphi_\GenC(\xvec|y) \; .
\end{equation}

Recalling the definition of our bins in Subsection~\ref{sec:bin_def}, we deduce that
\[
 \mathcal{R}\left( \varphi_\GenC(\xvec^*|y) \right)=\mathcal{R}\left( \varphi_\GenC(\xvec^*|y^*) \right)\geq\mathcal{R}\left( \varphi_\GenC(\xvec|y) \right)\;, 
\]
for all $y\in\mathcal{B}(z)$ and for all $\xvec\in\mathcal{X}^t$.
\end{proof}

\subsection{Properties of $\APPz$}
Recall that the APP measure $\APPz(\xvec|z)$ was defined in Subsection~\ref{subsec:psiSpecification}. We start this subsection by showing  that  $\APPz$ is ``close'' to the APP of the original channel.
\begin{lemm}
\label{lemma:eta_bound}
Let $\GenChannel$ be a generic $t$-user MAC, and let $\MergeAlph$ be the merged output alphabet conceived through applying the binning procedure to $\GenAlph$.
For each $z\in\MergeAlph$, let $\xvec^* = \xvec^*(z)$ be the leading-input defined by (\ref{eq:leading_input}), and let 
$\APPz(\xvec|z)$ be the probability measure on $\xvec\in\mathcal{X}^t$ defined in (\ref{eq:Merge_APP}).

Then for all $z\in\MergeAlph$ and $y\in\mathcal{B}(z)$,
\[
|\eta\left( \varphi_\GenC(\xvec|y) \right)-\eta\left( \APPz(\xvec|z) \right)|
\leq		\begin{cases}
				\;\frac{1}{\mu} 	&\text{if }\xvec\neq\xvec^* \;,\\
				\frac{q-1}{\mu}	&\text{if }\xvec=\xvec^*\;.
				\end{cases}
\] 
\end{lemm}

\begin{proof}
Consider a particular letter $y\in\mathcal{B}(z)$.
For all $\xvec\neq\xvec^*$, we have by (\ref{eq:Merge_App:notLeading}) that $\APPz(\xvec|z)$ belongs to the same quantization interval as 
$\varphi_\GenC(\xvec|y)$. Therefore, the first case is due to Corollary~\ref{cor:eta_diff_inBin}.

As for the second case, let $\{\Delta_i\}_{i=1}^M$ be as in Lemma~\ref{lemma:interval_width}. Also, for the leading region 
$i^*=i(\xvec^*)$, define the \emph{leading width} by
 \[
 \Delta^*=\Delta_{i^*} \;.
\] 
As Lemma~\ref{lemma:leading_region} declares the leading region to be the rightmost region in $\mathcal{L}(z)$, it follows from Lemma~\ref{lemma:interval_width} that either
\[
 i^*=M \;,\quad\text{or }\quad \Delta^*= \max\, \left\{  \Delta_i: i\in \mathcal{L}(z) \right\} \;.
\]
In words, the leading region is either the last region or the widest. 

Suppose first that $i^*<M$. Thus, the leading width is the largest. And so we claim that for all $\xvec\neq\xvec^*$,
\[
0 \leq \varphi_\GenC(\xvec|y)-\APPz(\xvec|z) \leq \Delta_i \leq \Delta^* \; ,
\]
where $i=i(\xvec)=\mathcal{R}\left( \varphi_\GenC(\xvec|y) \right)$. The leftmost inequality follows from (\ref{eq:Merge_App:notLeading}), while the middle follows from $\APPz(\xvec|z)$ and $\varphi_\GenC(\xvec|y)$ belonging to the same quantization interval. The rightmost inequality follows from our observation that $\Delta^*= \max\, \left\{  \Delta_i: i\in \mathcal{L}(z) \right\}$. 
Based on (\ref{eq:Merge_App:leading}), the above implies that
\begin{equation}
\label{eq:leading_APP_bound}
0 \leq \APPz(\xvec^*|z)-\varphi_\GenC(\xvec^*|y)\leq (q-1)\Delta^* \; .
\end{equation}
That is, $\xvec^*$ may have been ``pushed'' several regions higher: $\quad\mathcal{R}\left(\APPz(\xvec^*|z)\right)\geq\mathcal{R}\left(\varphi_\GenC(\xvec^*|y)\right)$.  However, Lemma~\ref{lemma:interval_width} assures that
$\Delta^*$ is no bigger than the width of subsequent regions. Thus
\[
\quad\mathcal{R}\left(\APPz(\xvec^*|z)\right)-\mathcal{R}\left(\varphi_\GenC(\xvec^*|y)\right) \leq q-1 \; ,
\] 
from which the second part of the lemma follows by induction, applying Lemma~\ref{lemma:eta_diff}.

If, on the other hand, $i^*=M$, then $\APPz(\xvec^*|z)$ must also belong to the last (and leading) region. The second part of the lemma follows then 
from Corollary~\ref{cor:eta_diff_inBin}.
\end{proof}

The quantity $\APPz(\xvec^*|z)$ frequently appears as a denominator. The main use of the following lemma is to show that such an expression is well defined.
\begin{lemm}
\label{lemma:APP_bound}
For $z\in\MergeAlph$, let $\xvec^* = \xvec^*(z)$ be the leading-input defined by (\ref{eq:leading_input}), and let 
$\APPz(\xvec|z)$ be the probability measure on $\xvec\in\mathcal{X}^t$ defined in (\ref{eq:Merge_APP}). Then,
\begin{equation}
\label{eq:psiLowerBound}
\APPz(\xvec^*|z) \geq \frac{1}{q} \; ,
\end{equation}
for all $z \in \MergeAlph$. 
\end{lemm}

\begin{proof}
Consider a given $z\in\MergeAlph$. Let the leading-output $y^*\in\mathcal{B}(z)$ be as in (\ref{eq:yleading}).
On the one hand, since the sum of $\varphi_\GenC(\xvec|y^*)$ over $\xvec\in\mathcal{X}^t$ is $1$, there exists a $\xvec\in\mathcal{X}^t$
such that
\begin{IEEEeqnarray}{l}
\label{eq:leading_APP_output}
\varphi_\GenC(\xvec|y^*)\geq \frac{1}{q} \; .
\end{IEEEeqnarray} 
On the other hand, by (\ref{eq:phiDoubleStar}), we have that
\begin{IEEEeqnarray*}{l}
\varphi_\GenC(\xvec^*|y^*)\geq \varphi_\GenC(\xvec|y^*) \; .
\end{IEEEeqnarray*} 
Thus,
\begin{equation}
\label{eq:APP_bound}
\APPz(\xvec^*|z) \geq \varphi_\GenC(\xvec^*|y^*) \geq \frac{1}{q} \; ,
\end{equation}
where the left inequality follows by (\ref{eq:leading_APP_bound}).
\end{proof}

Let $z \in \MergeAlph$ and $y\in\mathcal{B}(z)$ be given.  
We will shortly make use of the quantity
\begin{equation}
\label{eq:gamma}
\gamma(y)
\triangleq		 \frac{\varphi_\GenC(\xvec^*|y)}{\APPz(\xvec^*|z)} \; .
\end{equation}
Note that by (\ref{eq:psiLowerBound}) , $\gamma(y)$ is indeed well defined. Next, we claim that 
\begin{equation}
\label{eq:APP_uleading}
\APPz(\xvec^*|z) \geq \varphi_\GenC(\xvec^*|y) \geq \frac{1}{q}-\frac{1}{\mu} > 0\; .
\end{equation}
To justify this claim, note that the leftmost inequality follows from (\ref{eq:phiDoubleStar}) and (\ref{eq:APP_bound}). The middle inequality follows from (\ref{eq:quantization}) and (\ref{eq:APP_bound}) (recall that $\mathcal{R}\left(\varphi_\GenC(\xvec^*|y)\right)=\mathcal{R}\left(\varphi_\GenC(\xvec^*|y^*)\right)$ for all  $y\in\mathcal{B}(z)$). Finally, the rightmost inequality follows from (\ref{eq:mu_req}). 

Therefore,
\begin{equation}
\label{eq:gamma_range}
 0\leq \gamma (y) \leq 1 \;.
\end{equation}

Recall that by Lemma~\ref{lemma:eta_bound}, we have that  $\APPz$ is close to the APP of the original channel, $\varphi_\GenC$, in an additive sense (for large enough $\mu$). The following lemma states that $\APPz$ and $\varphi_\GenC$ are close in a multiplicative sense as well, when we are considering $\xvec^*$. The proof is given in the appendix.

\begin{lemm}
\label{lemma:gamma_bound}
Let $\GenChannel$ be a $t$-user MAC ,
 and let $\gamma(y)$ be given by (\ref{eq:gamma}). Then for all $y\in\GenAlph$,
\begin{equation}
\label{eq:gamma_bound}
0 \leq 1 - \frac{q(q-1)}{\mu}  
\leq \gamma(y) \leq 1 \;.
\end{equation}
\end{lemm}

\subsection{The MAC $\UpgC$} 
\label{sub:UpgChannel}
We now define the channel $\UpgChannel$, an upgraded version of $\GenChannel$. The definition makes heavy use of $\alpha_\xvec (y)$, defined in (\ref{eq:alpha}). Thus, as a first step, we prove the following Lemma.
\begin{lemm}
\label{lemm:alpha}
Let $\alpha_\xvec (y)$, be as in (\ref{eq:alpha}). Then, $\alpha_\xvec (y)$ is well defined and satisfies
\begin{equation}
\label{eq:alphaBounded}
0 \leq \alpha_\xvec (y) \leq 1 \; .
\end{equation}
\end{lemm}
\begin{proof}
The claim obviously holds if $\varphi_\GenC(\xvec|y) = 0$ due to (\ref{eq:alpha_exception}). So, we henceforth assume that $\varphi_\GenC(\xvec|y) > 0$, and thus have that
\begin{equation}
\label{eq:alphaInterestingCase}
\alpha_\xvec (y) = \frac{\APPz(\xvec|z)}{\varphi_\GenC(\xvec|y)} \cdot \frac{\varphi_\GenC(\xvec^*|y)}{\APPz(\xvec^*|z)} \; .
\end{equation}

By assumption, the first denominator is positive. Also, by (\ref{eq:psiLowerBound}), the second denominator is positive, and thus $\alpha_\xvec (y)$ is indeed well defined.

We now consider two cases. If $\xvec=\xvec^*$, then $\alpha_\xvec (y) = 1$, and the claim is obviously true. Thus, assume that $\xvec \neq \xvec^*$. Since we are dealing with probabilities, we must have that $\alpha_\xvec (y) \geq 0$. Consider the two fractions on the RHS of (\ref{eq:alphaInterestingCase}). By (\ref{eq:Merge_App:notLeading}), the first fraction is at most $1$, and by (\ref{eq:APP_uleading}) the second fraction is at most $1$. Thus, $\alpha_\xvec (y)$ is at most $1$.
\end{proof}

We now define $\UpgChannel$, an upgraded version of $\GenC$. 
For all $y\in\GenAlph$ and for all $\xvec\in\mathcal{X}^t$, define
\begin{subequations}
\label{eq:UpgC}

\begin{equation}
\UpgC(y|\xvec)= \alpha_\xvec(y)\cdot\GenC(y|\xvec) \;. \label{eq:UpgC_y}
\end{equation}
Whereas, for all $\boost_\mathbf{v}\in\BoostSet$ and for all $\xvec\in\mathcal{X}^t$, define
\begin{equation}
\UpgC(\boost_\mathbf{v}|\xvec) = 
\begin{cases}
																	  \varepsilon_\xvec=\sum\limits_{y\in\GenAlph}^{} (1-\alpha_\xvec (y))\GenC(y|\xvec) &\text{if } \xvec=\mathbf{v} \;, \\
																		0																													&\text{otherwise.}	
																	 \end{cases}
\end{equation}
\end{subequations}


The following lemma states that $\UpgC$ is indeed an upgraded version of $\GenC$.
\begin{lemm}
\label{lemma:upgradation}
Let $\GenChannel$ be a $t$-user MAC, and let $\UpgChannel$ be the MAC obtained by the procedure above. Then, $\UpgC$ is well-defined and is upgraded with respect to $\GenC$.
That is,
\[
 \UpgC \succeq \GenC \;. 
\]
\end{lemm}

\begin{proof}
Based on Lemma~\ref{lemm:alpha}, it can be easily verified that $\UpgC$ is indeed well-defined. We define the following intermediate channel $\mathcal{P}:(\UpgAlph)\rightarrow\GenAlph$, and prove the lemma by showing that $\GenC$ is obtained by the concatenation of $\UpgC$ followed by $\mathcal{P}$.   
Define for all $y\in\GenAlph$ and for all $y'\in(\UpgAlph)$,
\[
\mathcal{P} (y|y')=		\begin{cases}
												1 &\text{if $y'=y\in\GenAlph$} \;, \\
												\frac{\left[1-\alpha_\xvec(y)\right]\cdot\GenC(y|\xvec)}{\varepsilon_{\mathbf{u}}}	&\text{if } y'=\boost_{\mathbf{u}}\in\BoostSet\,,\\
												0 &\text{otherwise.}
												\end{cases}
\]

Let $y\in\GenAlph$ and  $\xvec\in\mathcal{X}$ be given. Now consider the sum
\[
 \sum_{y'\in\UpgAlph} \UpgC(y'|\xvec)\cdot \mathcal{P}(y|y')=
 \UpgC(y|\xvec)\cdot 1 + \sum_{\boost\in\BoostSet} \UpgC(\boost|\xvec)\cdot \mathcal{P}(y|\boost) \;. 
\]
Consider first the case in which $\varepsilon_\xvec=0$. In this case, the sum term, in the RHS, is zero (see (\ref{eq:BoostSet})). Moreover, (\ref{eq:alpha_exception}) and (\ref{eq:epsilon}) imply that $\alpha_\xvec(y)=1$. And so we have, by (\ref{eq:UpgC_y}), that
\[
 \sum_{y'\in\UpgAlph} \UpgC(y'|\xvec)\cdot \mathcal{P}(y|y')=\GenC(y|\xvec) \;.
\]

Next, consider the case where $\varepsilon_\xvec>0$. We have that
\begin{IEEEeqnarray*}{rCl}
\sum_{y'\in\UpgAlph} \UpgC(y'|\xvec)\cdot \mathcal{P}(y|y') &=& \UpgC(y|\xvec) + \varepsilon_\xvec\cdot \mathcal{P}(y|\boost_\xvec)\\ 
&=&\alpha_\xvec(y)\GenC(y|\xvec) + \left[1-\alpha_\xvec(y)\right]\cdot\GenC(y|\xvec) \\
&=&\GenC(y|\xvec) \;.
\end{IEEEeqnarray*}
 
\end{proof}

A boost symbol carries perfect information about what was transmitted through the channel. We now bound from above the average probability of receiving a boost symbol. This result will be useful in the proof of Theorem~\ref{th:upgraded_MAC}, where we bound the sum-rate increment of our upgraded approximation.
\begin{lemm}
\label{lemma:epsilon_bound}
Let $\varepsilon_\xvec$ be given by (\ref{eq:epsilon}) for all $\xvec\in\mathcal{X}^t$. Then,
\[ 
\sum_{\xvec\in\mathcal{X}^t} \left(\mathbb{P}(\bfX=\xvec)\cdot \varepsilon_\xvec\right) \leq \frac{q(q-1)}{\mu} \;.
\] 
\end{lemm} 

\begin{proof}
By definition (\ref{eq:epsilon}), we have that
\begin{IEEEeqnarray*}{rCl}
\sum_{\xvec\in\mathcal{X}^t} \mathbb{P}(\bfX=\xvec)\cdot\varepsilon_\xvec \;&=&\; \sum_{\xvec\in\mathcal{X}^t} \left[\, \mathbb{P}(\bfX=\xvec)\cdot \sum\limits_{y\in\GenAlph}^{}
 (1-\alpha_\xvec (y))\GenC(y|\xvec) \,\right] \\
&=&1-\sum_{\xvec\in\mathcal{X}^t}\left[\, \mathbb{P}(\bfX=\xvec)\cdot  \sum\limits_{y\in\GenAlph}^{} \alpha_\xvec (y)\GenC(y|\xvec)\,\right] \;. \label{eq:epsilon_sum}\IEEEyesnumber
\end{IEEEeqnarray*}
We now bound the second term. We have that 
\begin{IEEEeqnarray*}{rCl}
\sum_{\xvec\in\mathcal{X}^t}\left[\, \mathbb{P}(\bfX=\xvec)\cdot  \sum\limits_{y\in\GenAlph}^{} \alpha_\xvec (y)\GenC(y|\xvec)\,\right] &\;=\;&
 \sum\limits_{y\in\GenAlph}^{}\sum_{\substack{ \xvec\in\mathcal{X}^t:\\ \GenC(y|\xvec)>0}} \mathbb{P}(\bfX=\xvec)\cdot   \alpha_\xvec (y) \cdot \GenC(y|\xvec) \\&=&
\sum_{z\in\MergeAlph}\sum_{y\in\mathcal{B}(z)}\sum_{\substack{ \xvec\in\mathcal{X}^t:\\ \varphi_\GenC(\xvec|y)>0}} \mathbb{P}(\bfX=\xvec)\cdot \frac{\APPz(\xvec|z)}{\varphi_\GenC(\xvec|y)}\cdot\gamma(y)\cdot
			\GenC(y|\xvec)\\&\geq&
		\left( 1-\frac{q(q-1)}{\mu}\right)\sum_{z\in\MergeAlph}\sum_{y\in\mathcal{B}(z)}\sum_{\substack{ \xvec\in\mathcal{X}^t:\\ \varphi_\GenC(\xvec|y)>0}}\APPz(\xvec|z)\cdot\frac{\mathbb{P}(\bfX=\xvec)\cdot \GenC(y|\xvec)}{\varphi_\GenC(\xvec|y)} \\&=&
			\left( 1-\frac{q(q-1)}{\mu}\right)\sum_{z\in\MergeAlph}\sum_{y\in\mathcal{B}(z)}\sum_{\xvec\in\mathcal{X}^t}\APPz(\xvec|z)\cdot p_\GenC(y) \\&=&
		\left( 1-\frac{q(q-1)}{\mu}\right)\sum_{z\in\MergeAlph} \sum_{y\in\mathcal{B}(z)} p_\GenC(y)\sum_{\xvec\in\mathcal{X}^t}\APPz(\xvec|z) \\&=&
		 1-\frac{q(q-1)}{\mu}\;, \label{eq:second_term}\IEEEyesnumber
		\end{IEEEeqnarray*}
where the inequality is due to Lemma~\ref{lemma:gamma_bound}, and the equality that follows it is due to the observation below. If $\varphi_\GenC(\xvec|y)=0$, then based on 
(\ref{eq:APP_uleading}), we have that $\xvec\neq\xvec^*$. Therefore, by (\ref{eq:Merge_App:notLeading}), $\varphi_\GenC(\xvec|y)=0$ implies that $\APPz(\xvec|z)=0$ as well. That in turn leads to our observation that
\begin{equation}
\label{eq:sum_APPz}
\sum_{\substack{ \xvec\in\mathcal{X}^t:\\ \varphi_\GenC(\xvec|y)>0}}\APPz(\xvec|z)=\sum_{\xvec\in\mathcal{X}^t}\APPz(\xvec|z)=1 \;.
\end{equation}
As the  second term of (\ref{eq:epsilon_sum}) is bounded by (\ref{eq:second_term}), the proof follows.
\end{proof}

\subsection{Consolidation}
In the previous section, we defined $\UpgChannel$ which is an upgraded version of $\GenChannel\,$. 
Note that the output alphabet of $\UpgC$ is \emph{larger} than that of $\GenC$, and our original aim was to \emph{reduce} the output alphabet size. We do this now by consolidating letters which essentially carry the same information. 

Consider the output alphabet $\UpgAlph$ of our upgraded MAC $\UpgC$, compared to the original output alphabet $\GenAlph$. Note that, while the output letters $y\in\GenAlph$ are the same output letters we started with, their APP values are \emph{modified} and satisfy the following.  
\begin{lemm}
\label{lemma:same_APP}
Let $\UpgChannel$ be the MAC defined in Subsection~\ref{sub:UpgChannel}. Then, all the output letters $y\in\mathcal{B}(z)$ have the same modified APP values (for each $\xvec\in\mathcal{X}^t$ separately).
Namely,
\[
 \varphi_{\UpgC}(\xvec|y)=\APPz(\xvec|z) \;, 
\]
for all $\xvec\in\mathcal{X}^t$, and for all $z\in\MergeAlph$ and $y\in\mathcal{B}(z)$.
\end{lemm} 
\begin{proof}
First consider the case where $\varphi_\GenC(\xvec|y)=0$. On the one hand, $\varphi_{\UpgC}(\xvec|y)=0$ by (\ref{eq:APP}) and (\ref{eq:UpgC_y}). On the other hand, 
(\ref{eq:APP_uleading}) implies that $\xvec\neq\xvec^*$, and thus $\APPz(\xvec|z)=0$ as well, by (\ref{eq:Merge_App:notLeading}).

Now assume $\varphi_\GenC(\xvec|y)>0$. In that case,
\begin{IEEEeqnarray*}{rCl}{\setlength{\IEEEnormaljot}{10000pt}}
\varphi_{\UpgC}(\xvec|y)&=&\frac{\mathbb{P}(\bfX=\xvec)\cdot\UpgC(y|\xvec)}{\sum\limits_{\mathbf{v}\in\mathcal{X}^t} \mathbb{P}(\bfX=\mathbf{v})\cdot\UpgC(y|\mathbf{v}) }\\ \\
&=& \frac{\mathbb{P}(\bfX=\xvec)\cdot\alpha_\xvec(y)\cdot\GenC(y|\xvec)}{\sum\limits_{\substack{\mathbf{v}\in\mathcal{X}^t:\\ \GenC(y|\mathbf{v})>0}} \mathbb{P}(\bfX=\mathbf{v})\cdot\alpha_\mathbf{v}(y)\cdot\GenC(y|\mathbf{v}) }\\ \\  
&=& \frac{\frac{\mathbb{P}(\bfX=\xvec)\cdot\GenC(y|\xvec)}{\varphi_\GenC(\xvec|y)}\cdot \gamma(y) \cdot  \APPz(\xvec|z) }
				 { \sum\limits_{\substack{\mathbf{v}\in\mathcal{X}^t:\\ \varphi_\GenC(\mathbf{v}|y)>0}} \frac{\mathbb{P}(\bfX=\xvec)\cdot\GenC(y|\mathbf{v})}{\varphi_\GenC(\mathbf{v}|y)}\cdot \gamma(y) \cdot  \APPz(\mathbf{v}|z)}\\ \\
&=& \frac{ \APPz(\xvec|z)}{\sum\limits_{\mathbf{v}\in\mathcal{X}^t}\APPz(\mathbf{v}|z)}\\ \\
&=& \APPz(\xvec|z) \;,
\end{IEEEeqnarray*}
where the fourth equality follows from (\ref{eq:sum_APPz}).
\end{proof}

We have seen in Lemma~\ref{lemma:same_APP} that with respect to $\UpgC$
, all the members of $\mathcal{B}(z)$ have the same APP values. 
As will be pointed in Lemma~\ref{lemma: equivalence_r} in the sequel, consolidating symbols with equal APP values results in an equivalent channel. Thus consolidating all the members of every bin $\mathcal{B}(z)$ to one symbol $z$  results in an \emph{equivalent} channel $\FinChannel$ defined by (\ref{eq:FinC}). Note that consolidation simply means mapping all the members of $\mathcal{B}(z)$ to $z$ with probability 1. 
Formally, we have for all $z\in\FinAlph$ and for all $\xvec\in\mathcal{X}^t$, 
\begin{equation}
\label{eq:consolidation}
 \FinC(z|\xvec)=\begin{cases}
									 \sum\limits_{y\in\mathcal{B}(z)} \UpgC(y|\xvec) & \text{if } z\in\MergeAlph \,,\\ \\
									 \UpgC(z|\xvec)														& \text{if } z\in\BoostSet  \,.
									 \end{cases}
 \end{equation}
Based on (\ref{eq:UpgC}), it can be easily shown that the alternative definition above agrees with the definition of $\FinChannel$ in (\ref{eq:FinC}).

The rest of this section is dedicated to proving Theorem~\ref{th:upgraded_MAC}. But before that, we address the equivalence of $\UpgC$ and $\FinC$ in Lemma~\ref{lemma: equivalence_r}, which is proved in the appendix. In essence, we claim afterward that due to this equivalence, showing that $\UpgC\succeq\GenC$ implies that $\FinC\succeq\GenC$.
 
\begin{lemm}
\label{lemma: equivalence_r}
Let $\GenChannel$ be a $t$-user MAC, and let $y_1,\ldots,y_r\in\GenAlph$ be $r$ letters of \emph{equal} APP values, for some positive integer $r$. That is, for all $\xvec\in\mathcal{X}^t$,
\begin{equation} \varphi(\xvec|y_i)=\varphi(\xvec|y_j) \;,\;\text{for all}\;1\leq i\leq j\leq r  \;.\label{eq: same_APP}\end{equation}
Now let $Q:\mathcal{X}^t\rightarrow\mathcal{Z}$ be the $t$-user MAC obtained by consolidating $y_1,\ldots,y_r$ to one symbol $z$. This would make the output alphabet
\[
 \mathcal{Z}= \GenAlph\setminus\{y_1,\ldots,y_r\} \cup \{z\} \;. 
\]
Then, $ \GenC\equiv Q$ (the MACs $\GenC$ and $Q$ are equivalent).
\end{lemm}
We have mentioned that equivalence of MACs is a transitive relation. Therefore, consolidating bin after bin we finally have  by induction that $\UpgC\equiv\FinC$.
\\
\begin{proof}[Proof of Theorem~\ref{th:upgraded_MAC}] 

We first prove part (\ref{th:upgraded_MAC:upgraded}) of the theorem, which claims that the approximation is well defined and upgraded with respect to $W$.
Since $\FinChannel$ is a result of applying consolidation on $\UpgChannel$, it follows that $\FinC$ is well defined as well.

According to Lemma~\ref{lemma:upgradation}, $\UpgC\succeq \GenC$. Since $\UpgC$ and $\FinC$ are equivalent, and since upgradation transitivity immediately follows from the definition, it follows that $\FinC\succeq\GenC$. 

We now move to part (\ref{th:upgraded_MAC:sumRateBound}) of the theorem, which concerns the sum-rate difference. Recall that the random variable $Y$ has been defined as the output of $\GenChannel$ when the input is $\bfX$. Similarly, define $Z'$ as the output of $\FinChannel$ when the input is $\bfX$.

To estimate the APPs for $\FinChannel$, we may use (\ref{eq:APP}) and (\ref{eq:consolidation}). First, consider a non-boost symbol $z\in\MergeAlph$. Then, for all $\xvec\in\mathcal{X}^t$,
\[
 \varphi_{\FinC}(\xvec|z)=\frac{\mathbb{P}(\bfX=\xvec)\cdot \sum_{y\in\mathcal{B}(z)} \UpgC(y|\xvec)}{p_{\FinC}(z)}=\frac{\sum_{y\in\mathcal{B}(z)}\varphi_{\UpgC}(\xvec|y)\cdot p_{\UpgC}(y)}{\sum_{\tilde{y}\in\mathcal{B}(z)} p_{\UpgC}(\tilde{y})}=\APPz(\xvec|z) \;, 
\]
where the last equality follows from Lemma~\ref{lemma:same_APP}. 
Second, consider a boost symbol $\kappa\in K$. Then, for all $\xvec\in\mathcal{X}^t$,
\[
 \varphi_{\FinC}(\xvec|\boost)\in\{ 0,1 \} \;.
\]

 Denote the entropy of the probability distribution defined in Section~\ref{subsec:psiSpecification} by
\begin{IEEEeqnarray}{rCl}
\label{eq:psiEntropy}
 \EntropyZ &=& \sum_{\xvec\in\mathcal{X}^t} \eta\left[\APPz(\xvec|z)\right] \;.  
\end{IEEEeqnarray}
Thus
\[
R(\FinC) \,=\,
 H(\bfX) -\sum_{z\in\MergeAlph} p_{\FinC}(z)\EntropyZ \\
	-\sum_{\boost\in\BoostSet} p_{\FinC}(\boost)H(\bfX|Z^{\,'}=\boost) \;.
\]
However, the last term is zero due to the following observation. Given that the output of the MAC $\FinC$ is $\boost_{\mathbf{v}}$ for some $\mathbf{v}\in\mathcal{X}^t$, the input $\bfX$ is known to be $\mathbf{v}$ (it is deterministic). Hence $H(\bfX|Z'=\boost_{\mathbf{v}})=0$ for all $\boost_{\mathbf{v}}\in\BoostSet$. Hence
\begin{IEEEeqnarray}{rCl}
\label{eq:FinCap}
R(\FinC)&=&
 H(\bfX) -\sum_{z\in\MergeAlph} p_{\FinC}(z)\EntropyZ \;.
\end{IEEEeqnarray}	

Next we define a new auxiliary quantity to ease the proof. But first, define the random variable $Z$ as the  letter in the merged output alphabet $\MergeAlph$ corresponding to $Y$. Namely, the realization $Z=z$ occurs whenever $Y$ is contained in $\mathcal{B}(z)$. The probability of that realization is  
\begin{equation}
\label{eq:pBz}
\Pz \triangleq \mathbb{P}(Z=z)= \sum_{y\in\mathcal{B}(z)} p_\GenC(y) \;.
\end{equation}
Note that the joint distribution $\Pz\cdot\APPz(\xvec|z)$ does \emph{not} necessarily  induce a true MAC
 (for instance, it may contradict the true distribution of $\bfX$). Nevertheless, we plug this joint distribution into the sum-rate expression, with due caution. In other words, we define a new quantity $J(\bfX;Z)$, which is a surrogate for mutual information. 
 Namely, define
\begin{IEEEeqnarray*}{rCl}
 J(\bfX;Z)&\triangleq& H(\bfX) - \sum_{z\in\MergeAlph} \Pz \cdot \EntropyZ\;\quad\qquad\qquad \\
								&=& H(\bfX) - \sum_{z\in\MergeAlph} \Pz \sum_{\xvec\in\mathcal{X}^t} \eta\left[\APPz(\xvec|z)\right]   \;, \label{eq:fakeCap}\IEEEyesnumber
\end{IEEEeqnarray*}
where $\EntropyZ$ is given by $(\ref{eq:psiEntropy})$.

Now, we would like to bound the increment in sum-rate. To this end, we prove two bounds and then sum. First, note that
\begin{IEEEeqnarray*}{rCl}
J(\bfX;Z)-R(\GenC)&=&
 \sum_{y\in\GenAlph} 				p_\GenC(y)\sum_{\xvec\in\mathcal{X}^t} \eta(\varphi_\GenC(\xvec|y))
															-\sum_{z\in\MergeAlph} 			\Pz\sum_{\xvec\in\mathcal{X}^t} \eta(\APPz(\xvec|z))\\
&=&\sum_{z\in\MergeAlph}\sum_{y\in\mathcal{B}(z)} 		p_\GenC(y)\sum_{\xvec\in\mathcal{X}^t} \left[\eta(\varphi_\GenC(\xvec|y))-\eta(\APPz(\xvec|z)) \right]\\ 
&\leq& \sum_{z\in\MergeAlph}\sum_{y\in\mathcal{B}(z)} 		p_\GenC(y)\cdot|\eta(\varphi_\GenC(\xvec^*|y))-\eta(\APPz(\xvec^*|z))|\;+
\sum_{z\in\MergeAlph}\sum_{y\in\mathcal{B}(z)} 		p_\GenC(y)\sum_{\xvec\neq\xvec^*} |\eta(\varphi_\GenC(\xvec|y))-\eta(\APPz(\xvec|z))| \\
&\leq& 2\cdot\frac{q-1}{\mu} \;, \IEEEyesnumber \label{eq:first_bound}
\end{IEEEeqnarray*}
where the last inequality is due to Lemma~\ref{lemma:eta_bound}.

For the second bound, we subtract (\ref{eq:fakeCap}) from (\ref{eq:FinCap}) to get
\[
R(\FinC)-J(\bfX;Z)=
    \sum_{z\in\MergeAlph} (\Pz-p_{\FinC}(z))\EntropyZ \;.
\]
By (\ref{eq:FinC_1}), (\ref{eq:alphaBounded}), and (\ref{eq:pBz}), the parenthesized difference on the RHS is non-negative. Thus, 
\begin{IEEEeqnarray*}{rCl}
R(\FinC)-J(\bfX;Z)&\leq&
		\ln q \cdot\sum_{z\in\mathcal{Z}} (\Pz-p_{\FinC}(z)) 
		=\ln q \cdot\left[1-\sum_{z\in\mathcal{Z}} p_{\FinC}(z)\right]
		=\ln q \cdot \sum_{z\in\BoostSet} p_{\FinC}(z)
		=\ln q \cdot \sum_{\xvec\in\mathcal{X}^t} \left(\mathbb{P}(\bfX=\xvec)\cdot\varepsilon_\xvec\right) \;.
\end{IEEEeqnarray*}
Hence, by Lemma~\ref{lemma:epsilon_bound} we have a second bound:
		\begin{IEEEeqnarray}{l}
		 R(\FinC)-J(\bfX;Z) \leq \ln q\cdot 	\frac{q(q-1)}{\mu}\;. \label{eq:second_bound}
		\end{IEEEeqnarray}
The proof follows by adding the bounds (\ref{eq:first_bound}) and (\ref{eq:second_bound}).

Our last task is to prove part (\ref{th:upgraded_MAC:outputAlphabetBound}) of the theorem, which bounds the output alphabet size. Recall that $|\MergeAlph|$ is bounded by Lemma~\ref{lemma:num_bin}. Recalling that the number of boost symbols is bounded by $|\BoostSet|\leq|\mathcal{X}^t|=q$, the proof easily follows.
\end{proof}

\section*{Acknowledgments}
We thank Erdal Ar\i{}kan and Eren \c{S}a\c{s}o\u{g}lu for valuable comments.

\appendix
\begin{proof}[Proof of Lemma~\ref{lemma:eta_diff}]
Let $x$, $i$ and $\delta$ be as in Lemma~\ref{lemma:eta_diff}. If $x$ is in the last region, then the lemma simply follows from the definition in (\ref{eq:quantization}). So, suppose $i<M$, and let
\begin{equation}
\label{eq:Delta}
 \Delta=b_{i+1}-b_i \leq \frac{1}{\mu} \;, 
\end{equation}
where the inequality follows from (\ref{eq:quantization}).

We now consider two cases. If $\frac{1}{e^2}\leq x\leq 1$, then $|\eta'(p)|\leq 1$. Thus, by the triangle inequality,
\[
 |\eta(p+\delta)-\eta(p)|\leq \int_p^{p+\delta} |\eta'(\xi)|\,d\xi\leq \delta\leq \Delta=\frac{1}{\mu} \;, 
 \]
where the equality follows by part (ii) of Lemma~\ref{lemma:quantization}. 

In the other case left to consider, $0\leq x< \frac{1}{e^2}$. Recall that $\mu\geq 5$ by the assumption made in (\ref{eq:mu_req}). Hence
\[
\frac{1}{\mu}<\frac{1}{e}-\frac{1}{e^2} \; ,
\]
which implies that
\[
 x+\delta\leq x+\Delta\leq x+\frac{1}{\mu}<\frac{1}{e} \;. 
\]
Hence, the derivative function $\eta'$ is positive in the range $[p,p+\Delta]$. 
By the definition of $\Delta$ in (\ref{eq:Delta}), we have that the point $x+\Delta$ belongs to another region:
\[
 b_{i+1}\leq x+\Delta < \frac{1}{e} \;.
 \]
Thus, since $\eta$ is strictly increasing in $[0,\frac{1}{e})$,
\begin{IEEEeqnarray*}{rCl}
|\eta(p+\delta)-\eta(p)| &=&    \eta(p+\delta)-\eta(p)  \\
												 &\leq& \eta(p+\Delta)-\eta(p)	\\
												 &=& 			\left[\eta(b_{i+1})-\eta(p)\right]+\left[\eta(p+\Delta)-\eta(b_{i+1})\right] \;.
\end{IEEEeqnarray*} 
Hence, by the fundamental theorem of calculus,
\[
|\eta(p+\delta)-\eta(p)| \leq  \int_{p}^{b_{i+1}} \eta'(\xi)\,d\xi+\int_{b_{i+1}}^{x+\Delta} \eta'(\xi)\,d\xi	\;.		
\]
Since $\eta'(p)$ is a strictly decreasing function of $p$, the second integral can be upper-bounded by   
\[
\int_{b_{i+1}}^{p+\Delta} \eta'(\xi)\,d\xi < \int_{b_{i+1}-\Delta}^{p} \eta'(\xi)\,d\xi	\;.	
\] 
By (\ref{eq:Delta}), we have that $b_{i+1}-\Delta=b_i \;$.Thus, 
\begin{IEEEeqnarray*}{rCl}
|\eta(p+\delta)-\eta(p)|  &\leq&  \int_{p}^{b_{i+1}} \eta'(\xi)\,d\xi+\int_{b_{i}}^{p} \eta'(\xi)\,d\xi\\
											    &=&	   \eta(b_{i+1})-\eta(b_i)\leq\frac{1}{\mu}\;,		
\end{IEEEeqnarray*} 
where the last inequality follows from (\ref{eq:quantization}).



\end{proof}

\begin{proof}[Proof of Lemma~\ref{lemma:interval_width}]
Let us look at two quantization intervals $i$ and $j$, where $1\leq i<j< M$. Our aim is to prove that $\Delta_i \leq \Delta_j$. Consider first the simpler case in which $\Delta_j = \nicefrac{1}{\mu}$. Recall from (\ref{eq:quantization}) that $\nicefrac{1}{\mu}$ is an upper bound on the length of any interval, and specifically on $\Delta_i$. Thus, in this case, $\Delta_i \leq \Delta_j$.

Now, let us consider the case in which $\Delta_j < \nicefrac{1}{\mu}$. Thus, by (\ref{eq:quantization}), we must have that
\begin{equation}
\label{eq:etaJDiffTowardsContradiction}
\eta(b_{j+1}) - \eta(b_j) = \frac{1}{\mu} \; .
\end{equation}
We will now assume to the contrary that $\Delta_j < \Delta_i$, and show a contradiction to (\ref{eq:etaJDiffTowardsContradiction}).

Since $\Delta_j < \nicefrac{1}{\mu}$, we must have by part (\ref{lemma:quantization:it:otherwise}) of Lemma~\ref{lemma:quantization} that $b_j < \frac{1}{e^2}$. Since every interval length is at most $\nicefrac{1}{\mu}$, we must have that $\Delta_i \leq \nicefrac{1}{\mu}$. By the above, and recalling the assumption in (\ref{eq:mu_req}) that $\mu \geq 5$, we deduce that
\[
b_j + \Delta_j < b_j + \Delta_i \leq b_j + \frac{1}{\mu} < \frac{1}{e^2} + \frac{1}{\mu} < \frac{1}{e} \; .
\]
Thus, since $\eta'(p)$ is positive for $p < \frac{1}{e}$,
\[
\eta(b_{j+1})-\eta(b_j)	=\int_{b_j}^{b_j+\Delta_j} \eta'(p)\,dp 
											  <\int_{b_j}^{b_j+\Delta_i} \eta'(p)\,dp \; . 
\]
Now, since $b_i < b_j$ and $\eta'(p)$ is a strictly decreasing function of $x$, we have that
\[
\int_{b_j}^{b_j+\Delta_i} \eta'(p)\,dp < \int_{b_i}^{b_i+\Delta_i} \eta'(p)\,dp = \eta(b_{i+1}) - \eta(b_i) \; .
\]
Lastly, since $b_j < \frac{1}{e^2}$, we have that $b_{i+1} < \frac{1}{e^2}$. Thus, by part (\ref{lemma:quantization:it:firstCase}) of Lemma~\ref{lemma:quantization} we have that
\[
\eta(b_{i+1})-\eta(b_i)=\frac{1}{\mu} \;.
\] 
From the last three displayed equations, we deduce that
\[
\eta(b_{j+1})-\eta(b_j) < \frac{1}{\mu} \; ,
\]
which contradicts (\ref{eq:etaJDiffTowardsContradiction}). 
\end{proof}

\begin{proof}[Proof of Lemma~\ref{lemma:gamma_bound}]
We already know that $\gamma(y) \leq 1$, by (\ref{eq:gamma_range}). Thus, we now prove the lower bound on $\gamma(y)$.
To this end, we have by (\ref{eq:quantization}) and (\ref{eq:leading_APP_bound}) that for all $z\in\MergeAlph$ and $y\in\mathcal{B}(z)$,
\[
\APPz(\xvec^*|z)-\varphi_\GenC(\xvec^*|y)\leq (q-1)\cdot\frac{1}{\mu} \; .
\]
By (\ref{eq:psiLowerBound}), we can divide both sides of the above by $\APPz(\xvec^*|z)$ and retain the inequality direction. The result is
\begin{IEEEeqnarray*}{rCl}
\label{eq:leading_ratio}
 \frac{\varphi_\GenC(\xvec^*|y)}{\APPz(\xvec^*|z)} &\geq& 1-(q-1)\cdot\frac {\nicefrac{1}{\mu}}
																																					  {\APPz(\xvec^*|z)}  \\
																									 &\geq& 1 - \frac{q(q-1)}{\mu} \; , 
\end{IEEEeqnarray*}
where the last inequality yet again follows from (\ref{eq:psiLowerBound}). Thus, we have proved the lower bound on $\gamma(y)$ as well. 
Since, by our assumption in (\ref{eq:mu_req}), $\mu\geq q(q-1)$, the lower bound is indeed non-negative. 
\end{proof}

\begin{proof}[Proof of Lemma~\ref{lemma: equivalence_r}]
Let $\GenC$, $Q$ and $y_1,\ldots,y_r$ be as in Lemma~\ref{lemma: equivalence_r}. We would like to show that $\GenC$ and $Q$ satisfy both
\[
 Q\preceq\GenC \;\text{ and }\; Q\succeq\GenC \;.
 \]
It is obvious that $Q$ is degraded with respect to $\GenC$. This is because $Q$ is obtained from $\GenC$ by mapping with probability 1 one letter to another. The letters $y_1,\ldots,y_r$ are mapped into $z$, whereas the rest of the letters in $\GenAlph$ are mapped to themselves.

We must now show that $Q:\mathcal{X}^t\rightarrow\mathcal{Z}$ is upgraded with respect to $\GenChannel$. Namely, we must furnish an intermediate channel  
$\mathcal{P}:\MergeAlph\rightarrow\GenAlph$. Denote
\begin{IEEEeqnarray*}{l}
 a_i(\xvec) \triangleq \GenC(y_i|\xvec)= \frac{p_\GenC(y_i) \varphi_\GenC(\xvec|y_i)}{\mathbb{P}(\bfX=\xvec)\cdot} \;,\\ 
 A(\xvec)   \triangleq Q(z|\xvec)  = \sum_{1\leq i\leq r} a_i(\xvec) \;,
\end{IEEEeqnarray*}
 for all $\xvec\in\mathcal{X}^t$.  Note that by our running assumption on non-degenerate output letters,  $A(\tilde{\xvec}) > 0$ for some $\tilde{\xvec}\in\mathcal{X}^t$.
So let
\[
e_i \triangleq \frac{a_i(\tilde{\xvec})}{A(\tilde{\xvec})} \;.
\]
 Given (\ref{eq: same_APP}), we get that 
  \[
	e_i \cdot A(\xvec)=a_i(\xvec)
	\]
 for all $\xvec\in\mathcal{X}^t\,$. Hence we define for all $y\in\GenAlph$ and $s\in\MergeAlph$, 
\[
\mathcal{P}(y|s)= \begin{cases}
									e_i					&\text{if } (y,s)=(y_i,z) \;\text{for some}\; 1\leq i\leq r \;,\\
									1						&\text{if } y=s \;,\\
									0						&\text{otherwise.}
									\end{cases}
\]
Trivial algebra finishes the proof.
\end{proof}

\twobibs{
\bibliographystyle{IEEEtran}
\bibliography{\bibfilePath}
}
{
\ifdefined\bibstar\else\newcommand{\bibstar}[1]{}\fi

}

\end{document}